\documentclass[11pt]{article}
\usepackage{mathrsfs,amssymb,amsmath,amsthm,amsfonts}
\usepackage{graphicx}
\usepackage{subfig}
\usepackage{float}

\usepackage{color}
\usepackage[colorlinks, linkcolor=red, anchorcolor=green, citecolor=blue]{hyperref}

\allowdisplaybreaks

\newtheorem{theorem}{Theorem}[section]
\newtheorem{lemma}[theorem]{Lemma}

\newtheorem{proposition}[theorem]{Proposition}
\newtheorem{remark}[theorem]{Remark}
\newtheorem{definition}[theorem]{Definition}
\newtheorem{question}[theorem]{Question}
\newtheorem{conjecture}[theorem]{Conjecture}
\numberwithin{equation}{section}
\begin{document}

\title{\Large\bf Truncated Sparse Approximation Property and Truncated $q$-Norm Minimization
\footnotetext{\hspace{-0.35cm}\endgraf Corresponding author: lipengmath@126.com (P. Li).
}}

\author{Wengu Chen and Peng Li\,$^\ast$}

\maketitle
\textbf{Abstract.} This paper considers approximately sparse signal and low-rank matrix's recovery via truncated norm minimization $\min_{x}\|x_T\|_q$
and $\min_{X}\|X_T\|_{S_q}$ from noisy measurements. We first introduce truncated sparse approximation property, a more general robust null space property, and establish the stable recovery of signals and matrices under the truncated sparse approximation property. We also explore the relationship between the restricted isometry property and truncated sparse approximation property. And we also prove that if a measurement matrix $A$ or linear map $\mathcal{A}$ satisfies truncated sparse approximation property of order $k$, then the first inequality in restricted isometry property  of order $k$ and of order $2k$ can hold for certain different constants $\delta_{k}$ and $\delta_{2k}$, respectively. Last, we show that if $\delta_{t(k+|T^c|)}<\sqrt{(t-1)/t}$ for some $t\geq 4/3$, then measurement matrix $A$ and linear map $\mathcal{A}$ satisfy truncated sparse approximation property of order $k$. Which should point out is that when $T^c=\emptyset$, our conclusion implies that sparse approximation property of order $k$ is weaker than restricted isometry property of order $tk$.





\section{Introduction}\label{s1}

\hskip\parindent

The problem of sparse signal recovery naturally arises in genetics, communications and image processing. Prominent examples include DNA microarrays \cite{ES2005, PVMH2008}, wireless communications \cite{HS2009, THER2010}, magnetic resonance imaging \cite{LDP2007, VAHBPL2010},  and more. In such contexts, we often require to recover an unknown signal $x\in\mathbb{R}^n$ from an underdetermined system of linear equations
\begin{align}\label{systemequationsnoise}
b=Ax+z,
\end{align}
where $b\in\mathbb{R}^m$ are available measurements, the matrix $A\in\mathbb{R}^{m\times n}~(m<n)$ models the linear measurement process and
$z\in\mathbb{R}^m$ is a vector of measurement errors. The constrained $l_q$ minimization method estimates the signal $x$ by
\begin{align}\label{lqminimization}
\hat{x}=\arg\min_{x\in\mathbb{R}^n}\{\|x\|_q^q:~\text{subject~to}~b-Ax\in\mathcal{B}\}
\end{align}
where $0<q\leq 1$ and $\mathcal{B}$ is a set determined by the noise structure.  In  this paper, we consider two types of bounded noises \cite{CW2011}.
One is $l_p$ bounded noises \cite{DET2006}, i.e.,
\begin{align}\label{lpboundednoise}
\mathcal{B}^{l_p}(\eta)=\{z: \|z\|_p\leq\eta\}
\end{align}
for some constant $\eta$; and the other is motivated by \textit{Dantzig~Selector} procedure \cite{CT2007}, where
\begin{align}\label{Dantzigselectornoise}
\mathcal{B}^{DS}(\eta)=\{z: \|A^*z\|_\infty\leq\eta\}.
\end{align}
In particular, when $\mathcal{B}=\{0\}$, it is the noiseless case.

For the $l_q$-minimization problem (\ref{lqminimization}), there are many works under the restricted isometry property \cite{CT2005,CT2006,CS2008,CDD2009, CZ2013,CZ2013-1,CZ2014} and under the null space property \cite{DE2003, CDD2009, S2011, FR2013, F2014}. Let $x_k$ be the best $k$-sparse approximation vector of $x$ and $\sigma_k(x)_q$ be the error of best $k$-term approximation to the vector $x$ in $l_q$-norm.
Specially, we should point out that in 2011, Sun \cite{S2011} introduced the sparse approximation property
\begin{align*}
 \|x_k\|_r^q\leq \beta k^{q/r-1}\sigma_k(x)_q^q +D\|Ax\|_p^q
\end{align*}
where $0<q,p,r\leq\infty$, which was called sparse Riesz property in \cite{ST2016}.
And Sun \cite{S2011} showed that the sparse approximation property is an appropriate condition on a measurement matrix to consider stable recovery of any compressible
signal from its noisy measurements. We should point out that when $q=1$ and $p=2$, the sparse approximation property is equivalent to the $l_r$-robust null
space property \cite{FR2013,F2014}
\begin{align*}
\|x_K\|_r\leq \beta k^{1/r-1}\sigma_k(x)_1 +D\|Ax\|_2,
\end{align*}
where $K$ is any subset of $\subset \{1,\ldots,n\}$ with $|K|\leq k$. Here and what follows, $x_K$ is the vector equal to $x$ on $K$ and to zero on $K^c$; $K^c$ denotes the complementary set of $K$.

Although Basis Pursuit ((\ref{lqminimization}) for $q=1$) is much easy to solve, it requires significantly more measurements. In \cite{WY2010}, Wang and Yin proposed an iterative support detection (ISD) method that runs as fast as the best BP algorithms but requires significantly fewer measurements. From an incorrect reconstruction, support detection identifies an index set $I$ containing some elements of $\text{supp}(\bar{x})=\{j:x_j\neq 0\}$, and signal reconstruction solves
\begin{align}\label{TruncatedBP}
\min_x~\|x_T\|_1~\text{subject~to~}Ax=b,
\end{align}
where $T=I^c$. In fact, $I^{(s+1)}=\text{supp}(x^{(s)})$ in the $s$-th iteration of ISD algorithm.  In \cite{WY2010}, Wang and Yin  also generalized the null space property to the truncated null space property
\begin{align*}
\|v_K\|_1\leq\beta\|v_{T\cap K^c}\|_1 \end{align*}
where $T\subset\{1,\ldots,n\}$ is any set with $|T|=t$,  $K\subset T$ with $|K|\leq k$ and $v\in\mathcal{N}(A)$, and introduced an efficient implementation of ISD for recovering signals with fast decaying distributions of nonzeros from compressive sensing measurements.
Later, in her Dissertation, Zhang \cite{Z2013} considered the model (\ref{TruncatedBP}) in generalized case¡ª¡ªtruncated $l_q$-minimization with noise
\begin{align}\label{Truncatedq2}
\min_x~\|x_T\|_q^q~\text{subject~to~}\|Ax-b\|_2\leq\varepsilon,
\end{align}
where $0<q\leq 1$. Via restricted isometry property (RIP), she got a stable recovery of $x$ by (\ref{Truncatedq2}).

By the fact that the sparse approximation property is an additional strengthening of null space property in noisy case and is weaker than restricted isometry property,  therefore we consider the following question.

\begin{question} Whether can we introduce the truncated sparse approximation property for matrix $A$
and get a stable recovery of any compressible signal from its noisy measurements via the model
\begin{align}\label{Truncatedq}
\min_{x}~\|x_T\|_q^q ~~\text{subject~ to}~ ~b-Ax\in\mathcal{B},
\end{align}
where $0<q\leq 1$.
\end{question}
In this paper, we give affirmative answer for this question. We introduce the truncated sparse approximation property with $l_p$-norm constraint $\|b-Ax\|_p$
and Dantzig selector constraint $\|A^*(b-Ax)\|_{\infty}$, and get a stable recovery of signal $x$ via (\ref{Truncatedq}).

On the other hand, a closely related problem to compressed sensing is the low-rank matrix recovery,
which aims to recover an unknown low-rank matrix based on
its affine transformation
\begin{align}\label{Matrixsystemequationsnoise}
b=\mathcal{A}(X)+z,
\end{align}
where $X\in\mathbb{R}^{m\times n}$ is the decision variable and the linear map $\mathcal{A}: R^{m\times n}\rightarrow \mathbb{R}^l$, $z\in\mathbb{R}^l$ is the measurement error and vector $b\in\mathbb{R}^l$ is given. Low-rank matrices arise in an incredibly wide range of settings throughout science and applied mathematics. To name just a few examples, we commonly encounter low-rank matrices in contexts as varied as: ensembles
of signals \cite{DE2012,AR2015}, system identification \cite{LV2009}, adjacency matrices \cite{LLR1995}, distance matrices \cite{BLWY2006,BG2010,SY2007} and machine learning \cite{AEP2008,OTJ2010,D2004}.

Assume that the singular value decomposition of $X$ is $X=U\Lambda V^{T}=U\text{diag}(\lambda(X))V^{T}$ and $l=\min\{m,n\}$, where $\lambda(X)=(\lambda_1,\ldots,\lambda_{l})$ is the vector of singular values of the matrix $X$ with $\lambda_1\geq\lambda_2\geq\cdots\geq\lambda_{l}\geq 0$.
Let $\|X\|_{S_p}:=\|\lambda(X)\|_p=\big(\sum_{j=1}^{l}|\lambda_j(X)|^p\big)^{1/p}$ denote the Schatten $p$-norm of matrix $X$,  We should point out that $\|X\|_{S_1}=\|X\|_*$ the nuclear norm,
$\|X\|_{S_2}=\|X\|_F$ the Frobenius norm and $\|X\|_{S_{\infty}}=\|X\|_{2\rightarrow 2}$ the operator norm.
In 2010, Recht,  Fazel and Parrilo \cite{RFP2010} generalized the restricted isometry property (RIP) from vectors to matrices and showed that
if certain restricted isometry property holds for the linear transformation $\mathcal{A}$, the minimum-rank solution can be recovered by solving the
minimization of the nuclear norm
\begin{align*}
\min_{X}\|X\|_{*}~~\text{subject~to~}\mathcal{A}(X)=b. \end{align*}
Later, the low-rank matrix recovery problem has been studied by many scholars, readers can refer to
\cite{CP2011,CZ2013,CZ2013-1,CZ2014,CZ2015,ZL2017,ZL2017-1}
under matrix restricted isometry property and \cite{FR2013,KKRT2016} under rank null space property.

Now, a corresponding question is as follows.
\begin{question}
Whether can we introduce a truncated rank null space property, even a truncated rank sparse approximation property for
linear map $\mathcal{A}$ and obtain a stable recovery of the approximately low-rank matrix $X$ from its noisy measurements via the model
\begin{align}\label{MatrixTruncatedq}
\min_{X}~\|X_T\|_{S_q}^q ~~\text{subject~ to}~ ~b-\mathcal{A}(X)\in\mathcal{B},
\end{align}
where $\|X_T\|_{S_q}=\big(\sum_{j\in T}|\lambda_j(X)|^q\big)^{1/q}=\big(\sum_{j={l-t+1}}^l|\lambda_j(X)|^q\big)^{1/q}$ and $0<q\leq 1$.
\end{question}
Our answer to this question is also affirmative.

Our paper is organized as follows.
In Section \ref{s2}, we introduce truncated sparse approximation property for vector and matrix, respectively, and establish a stable recovery of approximately sparse signals and  low-rank matrices.
In Section \ref{s3}, we show that if a measurement matrix $A$ satisfies truncated sparse approximation property of order $k$, then the first inequality in restricted isometry property of order $k$ and of order $2k$ hold for certain different constants $\delta_{k}$ and $\delta_{2k}$, respectively.
And, we also show that if measurement matrix $A$ satisfies restricted isometry property with $\delta_{t(k+|T^c|)}<\sqrt{(t-1)/t}$, then $A$ also satisfies truncated sparse approximation property of order $k$. We note that when $T^c=\emptyset$, this conclusion implies that the $(l_2,l_1)-l_2$ sparse approximation property of order $k$ and and $(l_2,l_1)-\text{Dantzig~selector}$ sparse approximation property of order $k$ is weaker than restricted 2-isometry property of order $tk$.   And these relationships between the restricted isometry property and truncated sparse approximation property also hold for linear map $\mathcal{A}$. Finally, we summarize our conclusions and discuss the drawbacks of our results and file out some proper directions for coming study in Section 4.

\section{Truncated sparse approximation property and Stable Recovery \label{s2}}
\hskip\parindent

In this section, we will introduce $t$-truncated $(l_r,l_q)-l_p$ sparse approximation property and $t$-truncated $(l_r,l_q)-\text{Dantzig~selector}$ sparse approximation property of order $k$ for measurement matrix $A$, and get stable recovery of signal $x$ from (\ref{Truncatedq}) for $\mathcal{B}=\mathcal{B}^{l_p}(\eta)$ and $\mathcal{B}=\mathcal{B}^{DS}(\eta)$.  And we also introduce truncated rank sparse approximation property and establish the same stable recovery for approximately low-rank matrices.  Firstly, we consider the vector case.

\subsection{Vector Case \label{s2.1}}
\hskip\parindent

We introduce the following $t$-truncated sparse approximation property.

\begin{definition}\label{SparseRieszProperty}
Let $0<p,r\leq\infty$ and $0<q<\infty$. An $m\times n$ matrix $A$ satisfies $t$-truncated $(l_r,l_q)-l_p$ sparse approximation property of order $k$ with constants $D$ and $\beta$ if
\begin{align}\label{lpSparseRieszProperty}
\|x_{K}\|_r^q \leq D\|Ax\|_p^q + \beta k^{q/r-1}\sigma_k(x_{T})_q^q,
\end{align}
holds for all sets $T\subset\{1,\ldots,n\}$ with $|T|=t$, and all sets $K\subset T$-the index set of the $k$ largest (in magnitude) coefficients of $x_{T}$.

And an $m\times n$ matrix $A$ satisfies the $t$-truncated $(l_r,l_q)-\text{Dantzig~selector}$ sparse approximation property of order $k$ with constants $D$ and $\beta$ if
\begin{align}\label{DSSparseRieszProperty}
\|x_{K}\|_r^q \leq D\|A^*Ax\|_\infty^q + \beta k^{q/r-1}\sigma_k(x_{T})_q^q
\end{align}
holds for all sets $T\subset\{1,\ldots,n\}$ with $|T|=t$, and all sets $K\subset T$-the index set of the $k$ largest (in magnitude) coefficients of $x_{T}$.
\end{definition}

Next, we give out a proposition and a remark about truncated sparse approximation property.
\begin{proposition}
Let $0<p,r\leq\infty$ and $0<q<\infty$ and matrix $A$ satisfies $t$-truncated $(l_r,l_q)-l_p$ sparse approximation property of order $k$ with $t=|T|$, constants $D$ and $\beta$.
\begin{itemize}
\item[(1)]
  For any $\tilde{T}\subseteq T$ with $|\tilde{T}|=\tilde{t}\geq t$, matrix $A$ also satisfies $\tilde{t}$-truncated $(l_r,l_q)-l_p$ sparse approximation property of order $k$ with  constants $D$ and $\beta$.
\item[(2)]For any $\tilde{k}\leq k$, matrix $A$ also satisfies $t$-truncated $(l_r,l_q)-l_p$ sparse approximation property of order $\tilde{k}$ with $t=|T|$, constants $D$ and $\beta$.
\end{itemize}

And for $t$-truncated $(l_r,l_q)-\text{Dantzig~selector}$ sparse approximation property of order $k$, we have same results.
\end{proposition}
\begin{proof}
Assume that matrix $A$ satisfies $t$-truncated $(l_r,l_q)-l_p$ sparse approximation property of order $k$ with $t=|T|$, constants $D$ and $\beta$, then by the definition, we have
$$
\|x_{K}\|_r^q \leq D\|Ax\|_p^q + \beta k^{q/r-1}\sigma_k(x_{T})_q^q,
$$
holds for all sets $T\subset\{1,\ldots,n\}$ with $|T|=t$, and all sets $K\subset T$-the index set of the $k$ largest (in magnitude) coefficients of $x_{T}$.
For any  $|\tilde{T}|=\tilde{t}\geq t$ with $\tilde{T}\subseteq T$, one can get
$$
\|x_{K}\|_r^q \leq D\|Ax\|_p^q + \beta k^{q/r-1}\sigma_k(x_{T})_q^q\leq D\|Ax\|_p^q + \beta k^{q/r-1}\sigma_k(x_{\tilde{T}})_q^q,
$$
which implies that matrix $A$ also satisfies $\tilde{t}$-truncated $(l_r,l_q)-l_p$ sparse approximation property of order $k$ with  constants $D$ and $\beta$.

And for any $\tilde{k}\leq k$, let $\tilde{K}\subset T$-the index set of the $\tilde{k}$ largest (in magnitude) coefficients of $x_{T}$, then we have
\begin{align*}
\|x_{\tilde{K}}\|_r^q&\leq\|x_{K}\|_r^q\leq D\|Ax\|_P^q+\beta k^{q/r-1}\sigma_{k}(x_{T})_q^q
= D\|Ax\|_p^q + \beta \Big(\frac{\sigma_k(x_{T})_q}{k^{1/q-1/r}}\Big)^q\\
&\leq D\|Ax\|_p^q+\beta\Big(\frac{\sigma_{\tilde{k}}(x_{T})_q}{\tilde{k}^{1/q-1/r}}\Big)^q
=D\|Ax\|_p^q+\beta\tilde{k}^{q/r-1}\sigma_{\tilde{k}}(x_{T})_q^q.
\end{align*}

And for for $t$-truncated $(l_r,l_q)-\text{Dantzig~selector}$ sparse approximation property of order $k$, the proof is similar and we omit it here.
\end{proof}

\begin{remark}
If $x\in\mathcal{N}(A)$, then (\ref{lpSparseRieszProperty}) and (\ref{DSSparseRieszProperty}) become
\begin{align}\label{Nullspaceproperty}
\|x_{K}\|_r^q \leq \beta k^{q/r-1}\sigma_k(x_{T})_q^q.
\end{align}
When $r=q=1$, it is equivalent to the $t$-truncated null space property which was introduced in \cite{WY2010}.
\end{remark}

Next, we consider the model (\ref{systemequationsnoise}) in the setting where the observations contain noise and the signal is not exactly $k$-sparse.

\begin{theorem}\label{vectorstablerecovery}
Consider the signal model (\ref{systemequationsnoise}) with $\|z\|_p\leq\varepsilon$ and suppose $\hat{x}^{l_p}$ is the minimizer of (\ref{Truncatedq}) with
$\mathcal{B}=\mathcal{B}^{l_p}(\eta)$ defined in (\ref{lpboundednoise}) for some $\eta\geq\varepsilon$. Let $0<q\leq 1$, $q\leq r\leq\infty$, $1\leq p
\leq\infty$. If measurement matrix $A$ satisfies $t$-truncated $(l_r,l_q)-l_p$ sparse approximation property of order $k$ with
constants $D\in(0,\infty)$ and $\beta\in(0,1)$, then
\begin{align}\label{vectorstable1}
\|\hat{x}^{l_p}-x\|_r^q&\leq
\bigg(\max\Big\{\Big(\frac{|T^c|}{k}\Big)^q,1\Big\}+1+\Big(\frac{|T|-k}{k}\Big)^{q/r}\bigg)\frac{D}{1-\beta}(\varepsilon+\eta)^q\nonumber\\
&\hspace*{12pt}+\bigg(\max\Big\{\Big(\frac{|T^c|}{k}\Big)^q,1\Big\}+1+\Big(\frac{|T|-k}{k}\Big)^{q/r}\bigg)\frac{2\beta}{1-\beta}
k^{q/r-1}\sigma_k(x_T)_q^q,
\end{align}
and
\begin{align}\label{vectorstable2}
\|\hat{x}^{l_p}-x\|_q^q
&\leq\bigg(\max\Big\{\Big(\frac{|T^c|}{k}\Big)^q,1\Big\}\Big(\frac{|T^c|}{k}\Big)^{1-q/r}+2\bigg)\frac{D}{1-\beta}k^{1-q/r}(\varepsilon+\eta)^q\nonumber\\
&\hspace*{12pt}+\bigg(\max\Big\{\Big(\frac{|T^c|}{k}\Big)^q,1\Big\}\Big(\frac{|T^c|}{k}\Big)^{1-q/r}\frac{2\beta}{1-\beta}+\frac{2(1+\beta)}{1-\beta}\bigg)
\sigma_k(x_T)_q^q
\end{align}
if $q<r$, and
\begin{align}\label{vectorstable3}
\|\hat{x}^{l_p}-x\|_q^q&\leq\bigg(\max\Big\{\Big(\frac{|T^c|}{k}\Big)^q,1\Big\}+2\bigg)\frac{D}{1-\beta}(\varepsilon+\eta)^q\nonumber\\
&\hspace*{12pt}+\bigg(\max\Big\{\Big(\frac{|T^c|}{k}\Big)^q,1\Big\}\frac{2\beta}{1-\beta}+\frac{2(1+\beta)}{1-\beta}\bigg)\sigma_k(x_T)_q^q
\end{align}
if $q=r$.

Consider the signal model (\ref{systemequationsnoise}) with $\|A^*z\|_{\infty}\leq\varepsilon$ and suppose $\hat{x}^{DS}$ is the minimizer of (\ref{Truncatedq}) with $\mathcal{B}=\mathcal{B}^{DS}(\eta)$ defined in (\ref{Dantzigselectornoise}) for some $\eta\geq\varepsilon$. Let $0<q\leq 1$ and $q\leq r\leq\infty$. If measurement matrix $A$ satisfies $t$-truncated $(l_r,l_q)-\text{Dantzig~selector}$ sparse approximation property of order $k$ with constants $D\in(0,\infty)$ and $\beta\in(0,1)$, then we have
\begin{align*}
\|\hat{x}^{DS}-x\|_r^q&\leq
\bigg(\max\Big\{\Big(\frac{|T^c|}{k}\Big)^q,1\Big\}+1+\Big(\frac{|T|-k}{k}\Big)^{q/r}\bigg)\frac{D}{1-\beta}(\varepsilon+\eta)^q\nonumber\\
&\hspace*{12pt}+\bigg(\max\Big\{\Big(\frac{|T^c|}{k}\Big)^q,1\Big\}+1+\Big(\frac{|T|-k}{k}\Big)^{q/r}\bigg)\frac{2\beta}{1-\beta}
k^{q/r-1}\sigma_k(x_T)_q^q,
\end{align*}
and
\begin{align*}
\|\hat{x}^{DS}-x\|_q^q
&\leq\bigg(\max\Big\{\Big(\frac{|T^c|}{k}\Big)^q,1\Big\}\Big(\frac{|T^c|}{k}\Big)^{1-q/r}+2\bigg)\frac{D}{1-\beta}k^{1-q/r}(\varepsilon+\eta)^q\nonumber\\
&\hspace*{12pt}+\bigg(\max\Big\{\Big(\frac{|T^c|}{k}\Big)^q,1\Big\}\Big(\frac{|T^c|}{k}\Big)^{1-q/r}\frac{2\beta}{1-\beta}+\frac{2(1+\beta)}{1-\beta}\bigg)
\sigma_k(x_T)_q^q
\end{align*}
if $q<r$, and
\begin{align*}
\|\hat{x}^{DS}-x\|_q^q&\leq\bigg(\max\Big\{\Big(\frac{|T^c|}{k}\Big)^q,1\Big\}+2\bigg)\frac{D}{1-\beta}(\varepsilon+\eta)^q\nonumber\\
&\hspace*{12pt}+\bigg(\max\Big\{\Big(\frac{|T^c|}{k}\Big)^q,1\Big\}\frac{2\beta}{1-\beta}+\frac{2(1+\beta)}{1-\beta}\bigg)\sigma_k(x_T)_q^q
\end{align*}
if $q=r$.
\end{theorem}

Before proving Theorem \ref{vectorstablerecovery}, we first state an auxiliary lemma.
\begin{lemma}\label{coneconstraint}
Suppose $x,\hat{x}\in\mathbb{R}^{n}$, $v=\hat{x}-x$. If $\|\hat{x}_T\|_{q}^q\leq \|x_T\|_{q}^q$, then we have
$$
\|v_{T\cap K^c}\|_{q}^q\leq2\sigma_k(x_T)_q^q+\|v_{T\cap K}\|_{q}^q,
$$
where $K$ is the index set of $k$ largest (in magnitude) coefficients of $x_{T}$.
\end{lemma}
\begin{proof}
Since $\|\hat{x}_T\|_{q}^q\leq \|x_T\|_{q}^q$, one can get
\begin{align*}
\|x_T\|_q^q&\geq\|\hat{x}_T\|_q^q=\|(v+x)_T\|_q^q=\|(v+x)_{T\cap K}\|_q^q + \|(v+x)_{T\cap K^c}\|_q^q\\
&\geq\big(\|x_{T\cap K}\|_q^q -\|v_{T\cap K}\|_q^q\big) +\big(\|v_{T\cap K^c}\|_q^q-\|x_{T\cap K^c}\|_q^q\big),
\end{align*}
therefore,
\begin{align*}
\|v_{T\cap K^c}\|_q^q&\leq \big(\|x_{T}\|_q^q-\|x_{T\cap K}\|_q^q+\|x_{T\cap K^c}\|_q^q\big)+\|v_{T\cap K}\|_q^q \nonumber\\
&\leq 2\sigma_k(x_T)_q^q+\|v_{T\cap K}\|_q^q.
\end{align*}
\end{proof}

\begin{proof}[Proof of Theorem \ref{vectorstablerecovery}]
We finish our proof via the following seven steps.

\textbf{Step 1}. Tube constraint and cone constraint

Let $v=\hat{x}^{l_p}-x$, we have
\begin{align}\label{lpbounded}
\|Av\|_p&\leq\|A\hat{x}-b\|_p+\|Ax-b\|_p\leq \eta+\varepsilon.
\end{align}

Let $K$ denote the index set of $k$ largest (in magnitude) coefficients of $x_{T}$. By Lemma \ref{coneconstraint}, we have
\begin{align}\label{coneconstraintequation}
\|v_{T\cap K^c}\|_q^q&\leq 2\sigma_k(x_T)_q^q+\|v_{T\cap K}\|_q^q.
\end{align}

\textbf{Step 2}. Partition the set $\{1,\ldots,n\}$.

Let $\{1,\ldots,n\}=T^c\cup T=T^c\cup (T\cap K)\cup(T\cap K^c)$, for any $q\leq\tilde{r}\leq r$, we have
\begin{align}\label{Comininginequality}
\|v\|_{\tilde{r}}^q&=\big(\|v_{T^c}\|_{\tilde{r}}^{\tilde{r}}+\|v_{T\cap K}\|_{\tilde{r}}^{\tilde{r}}+\|v_{T\cap K^c}\|_{\tilde{r}}^{\tilde{r}}\big)^{q/{\tilde{r}}}
\leq \|v_{T^c}\|_{\tilde{r}}^q + \|v_{T\cap K}\|_{\tilde{r}}^q+\|v_{T\cap K^c}\|_{\tilde{r}}^q.
\end{align}
Next, we estimate  $v_{T^c}$, $v_{T\cap K}$ and $v_{T\cap K^c}$, respectively.

\textbf{Step 3}. Estimate $v_{T\cap K}$.

By the $t$-truncated $(l_r,l_q)-l_p$ sparse approximation property, we have
\begin{align*}
\|v_{T\cap K}\|_r^q&\leq D\|Av\|_p^q + \beta k^{q/r-1}\sigma_k(v_T)_q^q\leq D\|Av\|_p^q + \beta k^{q/r-1}\|v_{T\cap K^c}\|_q^q \\
&\leq D(\varepsilon+\eta)^q+\beta k^{q/r-1}\big(\|v_{T\cap K}\|_q^q+2\sigma_k(x_T)_q^q\big) \\
&\leq D(\varepsilon+\eta)^q +\beta k^{q/r-1}\big(|T\cap K|^{1/q-1/r}\|v_{T\cap K}\|_r\big)^q+ 2\beta k^{q/r-1}\sigma_k(x_T)_q^q\\
&\leq D(\varepsilon+\eta)^q +\beta \|v_{T\cap K}\|_r^q+ 2\beta k^{q/r-1}\sigma_k(x_T)_q^q,
\end{align*}
where the third inequality follows by (\ref{coneconstraintequation}).
Therefore, we obtain
\begin{align}\label{TcapK}
\|v_{T\cap K}\|_r^q\leq\frac{D}{1-\beta}(\varepsilon+\eta)^q+ \frac{2\beta}{1-\beta} k^{q/r-1}\sigma_k(x_T)_q^q.
\end{align}
by $0<\beta<1$.

\textbf{Step 4}. Estimate $v_{T\cap K^c}$.

Combining (\ref{coneconstraintequation}) with (\ref{TcapK}) yields
\begin{align}\label{TcapKc1}
\|v_{T\cap K^c}\|_q^q&\leq \|v_{T\cap K}\|_q^q +2\sigma_k(x_T)_q^q\leq \big(|T\cap K|^{1/q-1/r}\|v_{T\cap K}\|_r\big)^q+2\sigma_k(x_T)_q^q \nonumber\\
&\leq k^{1-q/r}\Big(\frac{D}{1-\beta}(\varepsilon+\eta)^q+ \frac{2\beta}{1-\beta} k^{q/r-1}\sigma_k(x_T)_q^q\Big)+2\sigma_k(x_T)_q^q\nonumber\\
&=\frac{D}{1-\beta}k^{1-q/r}(\varepsilon+\eta)^q + \frac{2}{1-\beta}\sigma_k(x_T)_q^q.
\end{align}
And by (\ref{TcapK}), we have
\begin{align}\label{TcapKc2}
\|v_{T\cap K^c}\|_r^q&\leq\big(|T\cap K^c|\max_{j\in T\cap K^c}|v_j|^r\big)^{q/r}\nonumber\\
&\leq\big(|T\cap K^c|\frac{\sum_{j\in T\cap K}|v_j|^r}{|T\cap K|}\big)^{q/r}=\Big(\frac{|T|-k}{k}\Big)^{q/r}\|v_{T\cap K}\|_r^q\nonumber\\
&\leq \Big(\frac{|T|-k}{k}\Big)^{q/r}\Big(\frac{D}{1-\beta}(\varepsilon+\eta)^q+ \frac{2\beta}{1-\beta} k^{q/r-1}\sigma_k(x_T)_q^q\Big).
\end{align}

\textbf{Step 5}. Estimate $v_{T^c}$.

Our idea of proving this part comes from \cite[Theorem 3.3]{WY2010}. Let $K_1=T^c$, $K_2=T\cap K$ and $K_3=T\cap K^c$. Next, we deal with ``$|K_1|\leq k$" and ``$|K_1|>k$", respectively.

\textbf{Case 1}: $|K_1|\leq k$.

We can find $K'\subset K_2$ such that $|K_1\cup K'|=k$ and $|K'|=k-(n-t)$. Let $\tilde{T}=K^c\cup K'$, then  $|\tilde{T}|=|K^c|+|K'|=(n-k)+(k-(n-t))=t$, $(K_1\cup K')\subset \tilde{T}$ and
$(K_1\cup K')^c\cap \tilde{T}=T\cap K^c=K_3$, therefore from the $t$-truncated $(l_r,l_q)-l_p$ sparse approximation property, we have
\begin{align*}
\|v_{T^c}\|_r^q&\leq\|v_{K_1\cup K'}\|_r^q\leq D\|Av\|_p^q+\beta k^{q/r-1}\sigma_k(v_{\tilde{T}})_q^q\\
&\leq D\|Av\|_p^q+\beta k^{q/r-1}\|v_{K_3}\|_q^q\\
&\leq D(\varepsilon+\eta)^q+\beta k^{q/r-1}\|v_{T\cap K^c}\|_q^q.
\end{align*}

\textbf{Case 2}: $|K_1|> k$.

Let $K''\subset K_1$ denote the set of indices corresponding to the largest $k$ entries of $v_{K_1}$. Take $\tilde{\tilde{T}}=K_3\cup K''$, then $|\tilde{\tilde{T}}|=|K_3|+|K''|=t-k+k=t$, $K''\subset \tilde{\tilde{T}}$ and $K''\cap \tilde{\tilde{T}}=K_3$.
Then $t$-truncated $(l_r,l_q)-l_p$ sparse approximation property leads us to conclude that
\begin{align*}
\|v_{T^c}\|_r^q&=\|v_{K_1}\|_r^q\leq\Big(\frac{|K_1|}{k}\|v_{K''}\|_r\Big)^q\\
&\leq \Big(\frac{|K_1|}{k}\Big)^q\big(D\|Av\|_p^q+\beta k^{q/r-1}\sigma_k(v_{\tilde{\tilde{T}}})_q^q\big)\\
&\leq \Big(\frac{|K_1|}{k}\Big)^q\big(D(\varepsilon+\eta)^q+\beta k^{q/r-1}\|v_{T\cap K^c}\|_q^q\big).
\end{align*}

Combining \textbf{Case 1} and \textbf{Case 2}, and using (\ref{TcapKc1}), we have
\begin{align}\label{Tc}
\|v_{T^c}\|_r^q&\leq \max\Big\{\Big(\frac{|K_1|}{k}\Big)^q,1\Big\}\big(D(\varepsilon+\eta)^q+\beta k^{q/r-1}\|v_{T\cap K^c}\|_q^q\big) \nonumber\\
&\leq \max\Big\{\Big(\frac{|T^c|}{k}\Big)^q,1\Big\}
\bigg(D(\varepsilon+\eta)^q+ \beta k^{q/r-1}\Big(\frac{D}{1-\beta}k^{1-q/r}(\varepsilon+\eta)^q + \frac{2}{1-\beta}\sigma_k(x_T)_q^q\Big)\bigg) \nonumber\\
&\leq \max\Big\{\Big(\frac{|T^c|}{k}\Big)^q,1\Big\}\Big(\frac{D}{1-\beta}(\varepsilon+\eta)^q+\frac{2\beta}{1-\beta}k^{q/r-1}\sigma_k(x_T)_q^q\Big).
\end{align}

\textbf{Step 6}. Obtain conclusions for $\mathcal{B}=\mathcal{B}^{l_p}(\eta)$.

For $\tilde{r}=r$, an application (\ref{Tc}) (\ref{TcapK}) and (\ref{TcapKc2}) to (\ref{Comininginequality}) yields
\begin{align*}
\|v\|_r^q&\leq\|v_{T^c}\|_r^q + \|v_{T\cap K}\|_r^q+\|v_{T\cap K^c}\|_r^q\\
&\leq \max\Big\{\Big(\frac{|T^c|}{k}\Big)^q,1\Big\}\Big(\frac{D}{1-\beta}(\varepsilon+\eta)^q+\frac{2\beta}{1-\beta}k^{q/r-1}\sigma_k(x_T)_q^q\Big)\\
&\hspace*{12pt}+\Big(\frac{D}{1-\beta}(\varepsilon+\eta)^q+ \frac{2\beta}{1-\beta} k^{q/r-1}\sigma_k(x_T)_q^q\Big)\\
&\hspace*{12pt}+\Big(\frac{|T|-k}{k}\Big)^{q/r}\Big(\frac{D}{1-\beta}(\varepsilon+\eta)^q+ \frac{2\beta}{1-\beta} k^{q/r-1}\sigma_k(x_T)_q^q\Big)\\
&\leq\bigg(\max\Big\{\Big(\frac{|T^c|}{k}\Big)^q,1\Big\}+1+\Big(\frac{|T|-k}{k}\Big)^{q/r}\bigg)\frac{D}{1-\beta}(\varepsilon+\eta)^q\\
&\hspace*{12pt}+\bigg(\max\Big\{\Big(\frac{|T^c|}{k}\Big)^q,1\Big\}+1+\Big(\frac{|T|-k}{k}\Big)^{q/r}\bigg)\frac{2\beta}{1-\beta} k^{q/r-1}\sigma_k(x_T)_q^q,
\end{align*}
which obtain the inequality (\ref{vectorstable1}).

For $\tilde{r}=q<r$, by (\ref{Tc}), (\ref{TcapK}) and (\ref{TcapKc1}), we have
\begin{align*}
\|v\|_q^q&\leq\|v_{T^c}\|_q^q + \|v_{T\cap K}\|_q^q+\|v_{T\cap K^c}\|_q^q\\
&\leq \big(|T^c|^{1/q-1/r}\|v_{T^c}\|_r\big)^q + \big(k^{1/q-1/r}\|v_{T\cap K}\|_r\big)^{q}+\|v_{T\cap K^c}\|_q^q\\
&\leq |T^c|^{1-q/r} \max\Big\{\Big(\frac{|T^c|}{k}\Big)^q,1\Big\}
\Big(\frac{D}{1-\beta}(\varepsilon+\eta)^q+\frac{2\beta}{1-\beta}k^{q/r-1}\sigma_k(x_T)_q^q\Big)\\
&\hspace*{12pt}+k^{1-q/r}\Big(\frac{D}{1-\beta}(\varepsilon+\eta)^q+ \frac{2\beta}{1-\beta} k^{q/r-1}\sigma_k(x_T)_q^q\Big)\\
&\hspace*{12pt}+\Big(\frac{D}{1-\beta}k^{1-q/r}(\varepsilon+\eta)^q + \frac{2}{1-\beta}\sigma_k(x_T)_q^q\Big)\\
&=\bigg(\max\Big\{\Big(\frac{|T^c|}{k}\Big)^q,1\Big\}\Big(\frac{|T^c|}{k}\Big)^{1-q/r}+2\bigg)\frac{D}{1-\beta}k^{1-q/r}(\varepsilon+\eta)^q\\
&\hspace*{12pt}+\bigg(\max\Big\{\Big(\frac{|T^c|}{k}\Big)^q,1\Big\}\Big(\frac{|T^c|}{k}\Big)^{1-q/r}\frac{2\beta}{1-\beta}+\frac{2(1+\beta)}{1-\beta}\bigg)\sigma_k(x_T)_q^q,
\end{align*}
which finish the inequality (\ref{vectorstable2}).

Last, we turn our attention to prove inequality (\ref{vectorstable3}).
For $r=q$, applying (\ref{Tc}) (\ref{TcapK}) and (\ref{TcapKc1}) to (\ref{Comininginequality}), one can get
\begin{align*}
\|v\|_q^q&\leq\|v_{T^c}\|_q^q + \|v_{T\cap K}\|_q^q+\|v_{T\cap K^c}\|_q^q\\
&\leq\max\Big\{\Big(\frac{|T^c|}{k}\Big)^q,1\Big\}\Big(\frac{D}{1-\beta}(\varepsilon+\eta)^q+\frac{2\beta}{1-\beta}\sigma_k(x_T)_q^q\Big)\\
&\hspace*{12pt}+ \Big(\frac{D}{1-\beta}(\varepsilon+\eta)^q+ \frac{2\beta}{1-\beta} \sigma_k(x_T)_q^q\Big)
+ \Big(\frac{D}{1-\beta}(\varepsilon+\eta)^q + \frac{2}{1-\beta}\sigma_k(x_T)_q^q\Big)\\
&\leq\bigg(\max\Big\{\Big(\frac{|T^c|}{k}\Big)^q,1\Big\}+2\bigg)\frac{D}{1-\beta}(\varepsilon+\eta)^q\\
&\hspace*{12pt}+\bigg(\max\Big\{\Big(\frac{|T^c|}{k}\Big)^q,1\Big\}\frac{2\beta}{1-\beta}+\frac{2(1+\beta)}{1-\beta}\bigg)\sigma_k(x_T)_q^q.
\end{align*}

\textbf{Step 7.} Inequalities for Dantzig selector.

In the case of noise type $\mathcal{B}^{DS}$, define $v=\hat{x}^{DS}-x$.
we can obtain
\begin{align}\label{DSbounded}
\|A^*Av\|_\infty&\leq\|A*(A\hat{x}-b)\|_{\infty}+\|A^*(Ax-b)\|_\infty\leq\eta+\varepsilon
\end{align}
and (\ref{DSSparseRieszProperty})
instead of (\ref{lpbounded}) and (\ref{lpSparseRieszProperty}), respectively. The following steps are similar with the case of noise type $\mathcal{B}^{l_p}$, we omit it here.
\end{proof}

For $r=q$, we have a sufficient and necessary condition for guaranteeing truncated sparse approximation property, which is similar to the conclusion of classic robust $l_q$-null space property, see \cite[Chapter 4.3]{FR2013}. We omit the proof details.

\begin{theorem}\label{sufficient-necessary}
\begin{itemize}
\item[(1)]The $m\times n $ measurement matrix $A$  satisfies $t$-truncated $(l_q,l_q)-l_p$ sparse approximation property of order $k$ with $t=|T|$,
    constants $D\in(0,\infty)$ and $\beta\in(0,1)$ if and only if
\begin{align}\label{e5}
\|(y-x)_{T}\|_q^q&\leq\frac{1+\beta}{1-\beta}\big(\|y_T\|_q^q-\|x_T\|_q^q+2\sigma_{k}(x_T)_q^q\big)+\frac{2D}{1-\beta}\|A(y-x)\|_p^q
\end{align}
holds for all vectors $y,x\in\mathbb{C}^n$.
\item[(2)] The $m\times n $ measurement matrix $A$ satisfies $t$-truncated $(l_q,l_q)-\text{Dantzig~selector}$ sparse approximation property of order $k$ with $t=|T|$, constants $D\in(0,\infty)$ and $\beta\in(0,1)$ if and only if
\begin{align}\label{e5}
\|(y-x)_{T}\|_q^q&\leq\frac{1+\beta}{1-\beta}\big(\|y_T\|_q^q-\|x_T\|_q^q+2\sigma_k(x_T)_q^q\big)+\frac{2D}{1-\beta}\|A^*A(y-x)\|_{\infty}^q
\end{align}
holds for all vectors $y,x\in\mathbb{C}^n$.
\end{itemize}
\end{theorem}

Last, we will show that if there exists a stable recovery for compressible signal $x$, then the measurement matrix $A$ also posses the truncated sparse approximation property for $r=p$.

\begin{theorem}\label{SRP}
Let $0<q,p\leq\infty$. If the error between the given vector $x$ and the solution $\hat{x}$ of (\ref{Truncatedq}) satisfies
\begin{align}\label{e7}
\|\hat{x}-x\|_p^q\leq B_1(\varepsilon+\eta)^q+B_2k^{q/p-1}\sigma_k(x_T)_q^q,
\end{align}
where $B_1,~B_2$ are positive constants independent of $\varepsilon,~\eta$ and $x$, and $|T|=t$, then
\begin{align*}
\|x\|_p^q\leq B_1\|Ax\|_p^q+B_2k^{q/p-1}\sigma_k(x_T)_q^q,
\end{align*}
or
\begin{align*}
\|x\|_p^q\leq B_1\|A^*Ax\|_\infty^q+B_2k^{q/p-1}\sigma_k(x_T)_q^q,
\end{align*}
and hence $A$ satisfies the $t$-truncated $(l_p,l_q)-l_p$ sparse approximation property of order $k$ with constants $B_1$ and $B_2$, or the $t$-truncated $(l_p,l_q)-l_p$ sparse approximation property of order $k$ with constants $B_1$ and $B_2$.
\end{theorem}
\begin{proof}
We observe that zero vector is the solution of truncated $l_q$-minimization problem (\ref{Truncatedq}) with $\eta=0$ and $\varepsilon=\|Ax\|_p$
or $\varepsilon=\|A^*Ax\|_\infty$, and $b=Ax$ for any $x\in\mathbb{R}^n$. Therefore the truncated sparse approximation property follows from (\ref{e7}) immediately.
\end{proof}

\begin{remark}\label{Algorithm}
In \cite{WY2010}, Wang and Yin proposed an iterative support detection (ISD) method to compute truncated $\ell_1$ minimization (\ref{TruncatedBP}). And in her Dissertation, Zhang \cite{Z2013} proposed a $\ell_q$-iterative support detection method to compute truncated $\ell_q$ minimization (\ref{Truncatedq}) with $\mathcal{B}=\mathcal{B}^{l_2}(\eta)$. Now another important question is that wether there exists a algorithm to compute truncated $\ell_q$ minimization (\ref{Truncatedq}) with $\mathcal{B}=\mathcal{B}^{DS}(\eta)$ or $\mathcal{B}=\mathcal{B}^{l_q}(\eta)$ for $0<q\leq 1$? This is one direction of our future research.
\end{remark}


\subsection{Matrix Cases \label{s2.2}}
\hskip\parindent

Now, we consider the matrix case. First, we introduce $t$-truncated rank sparse approximation property as follows.

\begin{definition}\label{SparseRieszProperty}
Let $0<p,r\leq\infty$ and $0<q<\infty$. A linear map $\mathcal{A}$ satisfies $t$-truncated $(l_r,l_q)-l_p$ rank sparse approximation property of order $k$ with constants $D$ and $\beta$ if
\begin{align}\label{MatrixlpSparseRieszProperty}
\|X_{K}\|_{S_r}^q \leq D\|\mathcal{A}(X)\|_p^q + \beta k^{q/r-1}\sigma_k(X_{T})_q^q,
\end{align}
holds for all sets $T\subset\{1,\ldots,\min\{m,n\}\}$ with $|T|=t$, and all sets $K\subset T$-the index set of the $k$ largest (in magnitude) coefficients of $\lambda(X)_{T}$.

And a linear map $\mathcal{A}$ satisfies the $t$-truncated $(l_r,l_q)-\text{Dantzig~selector}$ rank sparse approximation property of order $k$ with constants $D$ and $\beta$ if
\begin{align}\label{MatrixDSSparseRieszProperty}
\|X_{K}\|_{S_r}^q \leq D\|\mathcal{A}^{*}\mathcal{A}(X)\|_{S_\infty}^q + \beta k^{q/r-1}\sigma_k(X_{T})_q^q
\end{align}
holds for all sets $T\subset\{1,\ldots,\min\{m,n\}\}$ with $|T|=t$, and all sets $K\subset T$-the index set of the $k$ largest (in magnitude) coefficients of $\lambda(X)_{T}$.
\end{definition}

Let's consider the matrix recovery model (\ref{Matrixsystemequationsnoise}) in the setting where the observations contain noise and matrix is not exactly $k$-low-rank.

\begin{theorem}\label{Matrixstablerecovery}
Consider the matrix model (\ref{Matrixsystemequationsnoise}) with $\|z\|_p\leq\varepsilon$ and suppose $\hat{X}^{l_p}$ is the minimizer of
(\ref{MatrixTruncatedq}) with $\mathcal{B}=\mathcal{B}^{l_p}(\eta)$ defined in (\ref{lpboundednoise}) for some $\eta\geq\varepsilon$. Let $0<q\leq 1$,
$q\leq r\leq\infty$, $1\leq p \leq\infty$. If linear mapping $\mathcal{A}$ satisfies $t$-truncated $(l_r,l_q)-l_p$ sparse approximation property of order $k$ with $t=|T|$, constants $D\in(0,\infty)$ and $\beta\in(0,1)$, then
\begin{align}\label{Matrixstable1}
\|\hat{X}^{l_p}-X\|_{S_r}^q&\leq
\bigg(\max\Big\{\Big(\frac{|T^c|}{k}\Big)^q,1\Big\}+1+\Big(\frac{|T|-k}{k}\Big)^{q/r}\bigg)\frac{D}{1-\beta}(\varepsilon+\eta)^q\nonumber\\
&\hspace*{12pt}+\bigg(\max\Big\{\Big(\frac{|T^c|}{k}\Big)^q,1\Big\}+1+\Big(\frac{|T|-k}{k}\Big)^{q/r}\bigg)\frac{2\beta}{1-\beta}
k^{q/r-1}\sigma_k(X_T)_q^q,
\end{align}
and
\begin{align}\label{Matrixstable2}
\|\hat{X}^{l_p}-X\|_{S_q}^q
&\leq\bigg(\max\Big\{\Big(\frac{|T^c|}{k}\Big)^q,1\Big\}\Big(\frac{|T^c|}{k}\Big)^{1-q/r}+2\bigg)\frac{D}{1-\beta}k^{1-q/r}(\varepsilon+\eta)^q\nonumber\\
&\hspace*{12pt}+\bigg(\max\Big\{\Big(\frac{|T^c|}{k}\Big)^q,1\Big\}\Big(\frac{|T^c|}{k}\Big)^{1-q/r}\frac{2\beta}{1-\beta}+\frac{2(1+\beta)}{1-\beta}\bigg)
\sigma_k(X_T)_q^q
\end{align}
if $q<r$, and
\begin{align}\label{Matrixstable3}
\|\hat{X}^{l_p}-X\|_{S_q}^q&\leq\bigg(\max\Big\{\Big(\frac{|T^c|}{k}\Big)^q,1\Big\}+2\bigg)\frac{D}{1-\beta}(\varepsilon+\eta)^q\nonumber\\
&\hspace*{12pt}+\bigg(\max\Big\{\Big(\frac{|T^c|}{k}\Big)^q,1\Big\}\frac{2\beta}{1-\beta}+\frac{2(1+\beta)}{1-\beta}\bigg)\sigma_k(X_T)_q^q
\end{align}
if $q=r$.

Consider the matrix model (\ref{Matrixsystemequationsnoise}) with $\|\mathcal{A}^*(z)\|_{S_\infty}\leq\varepsilon$ and suppose $\hat{x}^{DS}$ is the minimizer of
(\ref{MatrixTruncatedq}) with $\overline{\mathcal{B}}=\overline{\mathcal{B}}^{DS}(\eta):=\{z\in \mathbb{R}^{l}: \|\mathcal{A}^*(z)\|_{S_\infty}\leq\eta\}$ for some $\eta\geq\varepsilon$. Let $0<q\leq 1$ and $q\leq r\leq\infty$. If linear mapping $\mathcal{A}$ satisfies $t$-truncated $(l_r,l_q)-\text{Dantzig~selector}$ sparse approximation property of order $k$ with $t=|T|$, constants $D\in(0,\infty)$ and $\beta\in(0,1)$, then we have
\begin{align*}
\|\hat{X}^{DS}-X\|_{S_r}^q&\leq
\bigg(\max\Big\{\Big(\frac{|T^c|}{k}\Big)^q,1\Big\}+1+\Big(\frac{|T|-k}{k}\Big)^{q/r}\bigg)\frac{D}{1-\beta}(\varepsilon+\eta)^q\nonumber\\
&\hspace*{12pt}+\bigg(\max\Big\{\Big(\frac{|T^c|}{k}\Big)^q,1\Big\}+1+\Big(\frac{|T|-k}{k}\Big)^{q/r}\bigg)\frac{2\beta}{1-\beta}
k^{q/r-1}\sigma_k(X_T)_q^q,
\end{align*}
and
\begin{align*}
\|\hat{X}^{DS}-X\|_{S_q}^q
&\leq\bigg(\max\Big\{\Big(\frac{|T^c|}{k}\Big)^q,1\Big\}\Big(\frac{|T^c|}{k}\Big)^{1-q/r}+2\bigg)\frac{D}{1-\beta}k^{1-q/r}(\varepsilon+\eta)^q\nonumber\\
&\hspace*{12pt}+\bigg(\max\Big\{\Big(\frac{|T^c|}{k}\Big)^q,1\Big\}\Big(\frac{|T^c|}{k}\Big)^{1-q/r}\frac{2\beta}{1-\beta}+\frac{2(1+\beta)}{1-\beta}\bigg)
\sigma_k(X_T)_q^q
\end{align*}
if $q<r$, and
\begin{align*}
\|\hat{X}^{DS}-X\|_{S_q}^q&\leq\bigg(\max\Big\{\Big(\frac{|T^c|}{k}\Big)^q,1\Big\}+2\bigg)\frac{D}{1-\beta}(\varepsilon+\eta)^q\nonumber\\
&\hspace*{12pt}+\bigg(\max\Big\{\Big(\frac{|T^c|}{k}\Big)^q,1\Big\}\frac{2\beta}{1-\beta}+\frac{2(1+\beta)}{1-\beta}\bigg)\sigma_k(X_T)_q^q
\end{align*}
if $q=r$.
\end{theorem}

To prove the results of matrix recovery, we need following lemma, which $q=1$ comes from \cite[Lemma 2.3]{RFP2010}, and $0<q<1$ comes from \cite[Lemma
2.2]{KX2013} and \cite[Lemma 2.1]{ZHZ2013}.

\begin{lemma}\label{orthogonaldecomposition}
Let $0<q\leq 1$. Let $X,Y\in\mathcal{R}^{m\times n}$ be matrices with $X^{T}Y=O$ and $XY^{T}=O$,  then the following holds:
\begin{itemize}
\item[(1)]$\|X+Y\|_{S_q}^q=\|X\|_{S_q}^q+\|Y\|_{S_q}^q$;
\item[(2)]$\|X+Y\|_{S_q}\geq\|X\|_{S_q}+\|Y\|_{S_q}$.
\end{itemize}
\end{lemma}

In order to get the cone constraint for matrix's  Schatten norm, we also need the following lemma which was gave by Yue and So in \cite{YS2016} and Audenaert \cite{A2014}.

\begin{lemma}\label{perturbationinequality}
Let $X, Y\in\mathbb{R}^{m\times n}$ be given matrices. Suppose that $f: \mathbb{R}_{+}\rightarrow \mathbb{R}_{+}$is a concave function satisfying $f(0)=0$. Then, for any $k\in\{1,\ldots,\min\{m,n\}\}$, we have
$$\sum_{j=1}^k|f(\lambda_j(X))-f(\lambda_j(Y))|\leq\sum_{j=1}^kf(\lambda_j(X-Y)).$$
\end{lemma}

\begin{lemma}
Suppose $X,\hat{X}\in\mathbb{R}^{m\times n}$, $V=\hat{X}-X$. If $\|\hat{X}_T\|_{S_q}^q\leq \|X_T\|_{S_q}^q$, then we have
$$
\|V_{T\cap K^c}\|_{S_q}^q\leq2\sigma_k(\lambda(X)_T)_q^q+\|V_{T\cap K}\|_{S_q}^q,
$$
where $K$ is the index set of $k$ largest (in magnitude) coefficients of $\lambda(X)_{T}$
\end{lemma}
\begin{proof}
By $\|\hat{X}_T\|_{S_q}^q\leq \|X_T\|_{S_q}^q$ and Lemma \ref{orthogonaldecomposition}, one can get
\begin{align*}
\|\lambda(X)_T\|_q^q&=\|X_T\|_{S_q}^q\geq\|\hat{X}_T\|_{S_q}^q=\|(V+X)_T\|_{S_q}^q=\|\lambda(V+X)_{T\cap K}\|_{q}^q + \|\lambda(V+X)_{T\cap K^c}\|_{q}^q.
\end{align*}
Since $x\rightarrow|x|^q$ is concave on $\mathbb{R}_+$ for any $q\in(0, 1]$, by taking $f(\cdot)=(\cdot)^q$ in above Lemma \ref{perturbationinequality}, we immediately obtain
\begin{align*}
\|\lambda(X)_T\|_q^q&\geq\big(\|\lambda(X)_{T\cap K}\|_q^q -\|\lambda(V)_{T\cap K}\|_q^q\big) +\big(\|\lambda(V)_{T\cap K^c}\|_q^q-\|\lambda(X)_{T\cap K^c}\|_q^q\big),
\end{align*}
Therefore,
\begin{align*}
\|V_{T\cap K^c}\|_{S_q}^q&=\|\lambda(V)_{T\cap K^c}\|_{q}^q\leq \big(\|\lambda(X)_{T}\|_q^q-\|\lambda(X)_{T\cap K}\|_q^q+\|\lambda(X)_{T\cap K^c}\|_q^q\big)+\|\lambda(V)_{T\cap K}\|_q^q \nonumber\\
&\leq 2\sigma_k(\lambda(X)_T)_q^q+\|\lambda(V)_{T\cap K}\|_q^q=2\sigma_k(\lambda(X)_T)_q^q+\|V_{T\cap K}\|_{S_q}^q.
\end{align*}
\end{proof}

The proof of these results of matrix case are similar to the vector case and we omit them here.

Conversely, if there exists a stable recovery of approximately low-rank matrix $X$, then the linear map $\mathcal{A}$ also posses the truncated sparse approximation property of $r=p$.

\begin{theorem}\label{MatrixSRP}
Let $0<q,p\leq\infty$. If the error between the given matrix $X$ and the solution $\hat{X}$ of (\ref{MatrixTruncatedq}) satisfies
\begin{align}\label{e8}
\|\hat{X}-X\|_{S_p}^q\leq B_1(\varepsilon+\eta)^q+B_2k^{q/p-1}\sigma_k(X_T)_q^q,
\end{align}
where $B_1,~B_2$ are positive constants independent of $\varepsilon,~\eta$ and $x$, and $|T|=t$, then
\begin{align*}
\|X\|_{S_p}^q\leq B_1\|\mathcal{A}(X)\|_p^q+B_2k^{q/p-1}\sigma_k(X_T)_q^q,
\end{align*}
or
\begin{align*}
\|X\|_{S_p}^q\leq B_1\|\mathcal{A}^*\mathcal{A}(X)\|_{S_\infty}^q+B_2k^{q/p-1}\sigma_k(X_T)_q^q,
\end{align*}
and hence $\mathcal{A}$ satisfies the $t$-truncated $(l_p,l_q)-l_p$ rank sparse approximation property of order $k$ with constants $B_1$ and $B_2$, or the $t$-truncated $(l_p,l_q)-\text{Dantzig-selector}$ rank sparse approximation property of order $k$ with constants $B_1$ and $B_2$.
\end{theorem}

\section{Truncated sparse approximation property and Restricted Isometry Property\label{s3}}
\hskip\parindent

In this section, we will consider the relationship between the restricted isometry property and truncated sparse approximation property. First, we recall the definition of restricted $p$-isometry property, which was introduced in \cite{CT2005} for $p=2$ and in \cite{CS2008} for $0<p\leq 1$.

\begin{definition}\label{RIP}
For an matrix $A$, $k>0$, and $0 < p\leq1$ or $p=2$, the $k$-th restricted $p$-isometry constants $\delta_k=\delta_k(A)$ is the smallest numbers such that
\begin{align*}
(1-\delta_k)\|x\|_2^p\leq\|Ax\|_p^p\leq(1+\delta_k)\|x\|_2^p
\end{align*}
for all $x$ such that $\|x\|_0\leq k$.
\end{definition}

And the matrix-restricted $2$-isometry property was first introduced by Recht,  Fazel and  Parrilo in \cite{RFP2010}. Later, it was extended to
matrix-restricted $p$-isometry property for $0<p\leq 1$ by Zhang, Huang and Zhang in \cite{ZHZ2013}.

\begin{definition}\label{MatrixRIP}
Let $0 < p\leq1$ or $p=2$, $\mathcal{A} : \mathbb{R}^{m¡Á\times n}\rightarrow \mathbb{R}^l$ be a linear map. Without loss of generality,
assume $m \leq n$. For every integer $k$ with $1\leq k\leq \min\{m,n\}$, define the matrix $k$-th restricted $p$-isometry
constant to be the smallest number $\delta_k=\delta_k(\mathcal{A})$ such that
\begin{align*}
(1-\delta_k)\|X\|_{S_2}^p\leq\|\mathcal{A}(X)\|_p^p\leq(1+\delta_k)\|X\|_{S_2}^p
\end{align*}
holds for all matrices $X$ of rank at most $k$.
\end{definition}

\subsection{Truncated sparse approximation property implies the first inequality in Restricted Isometry Property\label{s3.1}}
\hskip\parindent

Firstly, we show that if a measurement matrix satisfies truncated sparse approximation property of order $k$, then the first inequality in restricted isometry property of order $k$ and of order $2k$ hold for certain different constants $\delta_{k}$ and $\delta_{2k}$, respectively.

\begin{theorem}\label{k2ksparsevector}
Let $0<q\leq r\leq\infty$, $0<p\leq\infty$.
\item[(1)]If a matrix $A$ satisfies $t$-truncated $(l_r,l_q)-l_p$ sparse approximation property of order $k$ with $t=|T|$, constants $D\in(0,\infty)$ and $\beta\in(0,1)$, then
for all $x\in\Sigma_k(T)=\{x\in\mathbb{R}^n:\|x_T\|_0\leq k\}$,
\begin{align}
\frac{1}{C_1}\|x\|_r^q\leq\|Ax\|_p^q,
\end{align}
where
$$
C_1=\bigg(\max\Big\{\Big(\frac{|T^c|}{k}\Big)^q,1\Big\}+1\bigg)D,
$$
and for all $x\in\Sigma_{2k}(T)$,
\begin{align}
\frac{1}{C_2}\|x\|_r^q\leq\|Ax\|_p^q,
\end{align}
where
$$
C_2=\Bigg(\frac{2}{1-\beta}+\max\bigg\{\Big(\frac{|T^c|}{k}\Big)^q,1\bigg\}
\bigg(\Big(\frac{|T\cap K^c|}{k}\Big)^{1-q/r}\frac{2\beta}{1-\beta}+1\bigg)\Bigg)D.
$$

\item[(2)]If a matrix $A$ satisfies $t$-truncated $(l_r,l_q)-\text{Dantzig-selsector}$ sparse approximation property of order $k$ with $t=|T|$, constants $D\in(0,\infty)$ and $\beta\in(0,1)$, then
for all $x\in\Sigma_k(T)$,
\begin{align}
\frac{1}{C_1}\|x\|_r^q\leq\|A^*Ax\|_\infty^q,
\end{align}
and for all $x\in\Sigma_{2k}(T)$,
\begin{align}
\frac{1}{C_2}\|x\|_r^q\leq\|A^*Ax\|_\infty^q,
\end{align}
where $C_1$ and $C_2$ are the constants in above (1).
\end{theorem}
\begin{proof}
We only prove part (1).
Let $K$ denotes the index set of the $k$ largest (in magnitude) coefficients of $x_T$.

(a). First, we deal with $k$-sparse vector.
Take a $k$-sparse vector $x_T\in \mathbb{R}^n$, then $x_T=x_{K}$ and $\sigma_k(x_T)_q=\|x_{T\cap K^c}\|_q=0$. From $t$-truncated $(l_r,l_q)-l_p$ sparse approximation property, one can get
$$
\|x_T\|_r^q=\|x_K\|_r^q\leq D\|Ax\|_p^q.
$$
Similar to the \textbf{Step 5} of the proof of Theorem \ref{vectorstablerecovery}, we can get
\begin{align*}
\|x_{T^c}\|_{r}^q&\leq\max\Big\{\Big(\frac{|T^c|}{k}\Big)^q,1\Big\}\big(D\|Ax\|_p^q+\beta k^{q/r-1}\|x_{T\cap K^c}\|_q^q\big)\\
&=\max\Big\{\Big(\frac{|T^c|}{k}\Big)^q,1\Big\}D\|Ax\|_p^q.
\end{align*}
Therefore,
\begin{align*}
\|x\|_r^q&\leq\|x_T\|_r^q+\|x_{T^c}\|_r^q\leq D\|Ax\|_p^q+\max\Big\{\Big(\frac{|T^c|}{k}\Big)^q,1\Big\}D\|Ax\|_p^q\\
&=\bigg(\max\Big\{\Big(\frac{|T^c|}{k}\Big)^q,1\Big\}+1\bigg)D\|Ax\|_p^q\\
&:=C_1\|Ax\|_p^q.
\end{align*}

(b). Next, we consider the $2k$-sparse vector. Take a $2k$-sparse vector $x_{T}\in\mathbb{R}^n$, we can write
$$
x_T=x_{K_1}+x_{K_2},
$$
where $K_1, K_2\subset T$ with $\|x_{K_j}\|_0\leq k$ for $j=1,2$ and $K_1\cap K_2=\emptyset$, then we have
\begin{align}\label{K1K2}
\|x_T\|_r^q\leq \|x_{K_1}\|_r^q+\|x_{K_2}\|_r^q.
\end{align}
The $t$-truncated $(l_r,l_q)-l_p$ sparse approximation property leads us to conclude that
\begin{align}\label{lrKiKj}
\|x_{K_i}\|_r^q\leq D\|Ax\|_p^q+\beta k^{q/r-1}\|x_{K_j}\|_q^q\leq D\|Ax\|_p^q+\beta\|x_{K_j}\|_r^q,
\end{align}
where $i,j\in\{1,2\}$ and $i\neq j$. Therefore, we can get
$$
\|x_{K_1}\|_r^q+\|x_{K_2}\|_r^q\leq 2D\|Ax\|_p^q+\beta(\|x_{K_1}\|_r^q+\|x_{K_2}\|_r^q),
$$
which implies that
\begin{align}\label{e9}
\|x_T\|_r^q\leq\|x_{K_1}\|_r^q+\|x_{K_2}\|_r^q\leq\frac{2D}{1-\beta}\|Ax\|_p^q.
\end{align}

On the other hand, similar to the \textbf{Step 5} of the proof of Theorem \ref{vectorstablerecovery}, we also can get
\begin{align*}
\|x_{T^c}\|_r^q\leq\max\Big\{\Big(\frac{|T^c|}{k}\Big)^q,1\Big\}\big(D\|Ax\|_p^q+\beta k^{q/r-1}\|x_{T\cap K^c}\|_q^q\big),
\end{align*}
and by (\ref{e9}), we get
\begin{align*}
\|x_{T\cap K^c}\|_q^q&\leq|T\cap K^c|^{1-q/r}\|x_{T\cap K^c}\|_r^q\leq|T\cap K^c|^{1-q/r}\|x_T\|_r^q\\
&\leq\frac{2D}{1-\beta}|T\cap K^c|^{1-q/r}\|Ax\|_p^q.
\end{align*}
Therefore, we can obtain a upper bound for $\|x_{T^c}\|_r^q$ as follows
\begin{align}\label{e10}
\|x_{T^c}\|_r^q&\leq\max\bigg\{\Big(\frac{|T^c|}{k}\Big)^q,1\bigg\}
\bigg(\Big(\frac{|T\cap K^c|}{k}\Big)^{1-q/r}\frac{2\beta}{1-\beta}+1\bigg)D\|Ax\|_p^q.
\end{align}

Therefore, combining (\ref{e9}) and (\ref{e10}), we can get
\begin{align*}
\|x\|_r^q&\leq\|x_T\|_r^q+\|x_{T^c}\|_r^q\\
&\leq\frac{2D}{1-\beta}\|Ax\|_p^q
+\max\bigg\{\Big(\frac{|T^c|}{k}\Big)^q,1\bigg\}
\bigg(\Big(\frac{|T\cap K^c|}{k}\Big)^{1-q/r}\frac{2\beta}{1-\beta}+1\bigg)D\|Ax\|_p^q\\
&\leq\Bigg(\frac{2}{1-\beta}+\max\bigg\{\Big(\frac{|T^c|}{k}\Big)^q,1\bigg\}
\bigg(\Big(\frac{|T\cap K^c|}{k}\Big)^{1-q/r}\frac{2\beta}{1-\beta}+1\bigg)\Bigg)D\|Ax\|_p^q\\
&:=C_2\|Ax\|_p^q.
\end{align*}
\end{proof}

\begin{remark}
When $p=2$ or $0<p\leq1$, and $r=2$, we can obtain from Theorem \ref{k2ksparsevector} that if measurement matrix $A$ has the $t$-truncated $(l_r,l_q)-l_p$ sparse approximation property or $t$-truncated $(l_r,l_q)-\text{Dantzig-selsector}$ sparse approximation property of order $k$ with $t=|T|$, constants $D\in(0,\infty)$ and $\beta\in(0,1)$, then the first inequality in the restricted isometry property of order $k$ and of order $2k$ hold for constants $\delta=1-1/{(C_1)^{p/q}}$ and $\delta=1-1/{(C_2)^{p/q}}$, respectively.
\end{remark}

\subsection{Restricted Isometry Property implies Truncated sparse approximation property \label{s3.2}}
\hskip\parindent

In this subsection, we will show that the truncated sparse approximation property of order $k$ can be deduced from the restricted 2-isometry property with $\delta_{tk}<\sqrt{(t-1)/t}$ for some $t\geq4/3$. And especially when $T^c=\emptyset$-the classic case, it shows us that the $(l_2,l_1)-l_2$ sparse approximation property of order $k$ and and $(l_2,l_1)-\text{Dantzig~selector}$ sparse approximation property of order $k$ is weaker than restricted 2-isometry property of order $tk$. Our main result can be stated as follows.

\begin{theorem}\label{RIPSRP}
If $A$ has restricted 2-isometry property of order $t(k+|T^c|)$ with $\delta_{t(k+|T^c|)}<\sqrt{(t-1)/t}$ for some $t\geq 4/3$, then
\begin{align*}
\|x_{K}\|_2&\leq\frac{2\sqrt{1+\delta}}{1-\delta^2}\|Ax\|_2+ \frac{\delta}{\sqrt{(t-1)(1-\delta^2)}}\frac{\sigma_k(x_T)_1}{\sqrt{k}}\\
&=:D_1\|Ax\|_2+\beta \frac{\sigma_k(x_T)_1}{\sqrt{k}}
\end{align*}
and
\begin{align*}
\|x_{K}\|_2&\leq\frac{2\sqrt{2(k+|T^c|)}}{1-\delta^2}\|A^*Ax\|_{\infty}+ \frac{\delta}{\sqrt{(t-1)(1-\delta^2)}}\frac{\sigma_k(x_T)_1}{\sqrt{k}}\\
&=:D_2\|A^*Ax\|_{\infty}+\beta \frac{\sigma_k(x_T)_1}{\sqrt{k}},
\end{align*}
where $T\subset\{1,\ldots,n\}$ is the any set with $|T|=t$, $K\subset T$ is the index set of the $k$ largest (in magnitude) coefficients of $x_T$.
\end{theorem}

\begin{proof}
Our proof adapts Fourcat's \cite[Theorem 5]{F2017} approach.
Let $K_0:=K\cup T^c$ and
$$
\alpha=\frac{\|x_{T\cap K^c}\|_1}{k+|T^c|},
$$
where $k+|T^c|=|K_0|$.
We partition the set $T\cap K^c$-complement of
$K\cup T^c$ as
$$
T\cap K^c=K'\cup K'',
$$
where
\begin{align*}
K'&:=\Big\{j\in T\cap K^c: |x_j|>\frac{\alpha}{t-1}\Big\},\\
K''&:=\Big\{j\in T\cap K^c: |x_j|\leq\frac{\alpha}{t-1}\Big\}.
\end{align*}
Denote that $k'=\|v_{K'}\|_0$, we can derive that $k'\leq (t-1)(k+|T^c|)$ from
$$
\alpha(k+|T^c|)=\|x_{T\cap K^c}\|_1\geq\|x_{K'}\|_1\geq k'\frac{\alpha}{t-1}.
$$
By
$$
\|x_{K''}\|_{\infty}\leq \frac{\alpha}{t-1}
$$
and
$$
\|x_{K''}\|_1=\|x_{T\cap K^c}\|_1-\|x_{K'}\|_1\leq \alpha(k+|T^c|)-\frac{k'\alpha}{t-1}=((t-1)(k+|T^c|)-k')\frac{\alpha}{t-1},
$$
we know that the vector $x_{K''}$ belongs to a scaled version of the polytope $\big((t-1)(k+|T^c|)-k'\big)\mathcal{B}^{l_1}(\alpha/(t-1))\cap \mathcal{B}^{l_\infty}(\alpha/(t-1))$. By sparse representation of a polytope in \cite[Lemma 1.1]{CZ2014}, this polytope can be represented as the convex hull of $\big((t-1)(k+|T^c|)-k'\big)$-sparse vectors:
$$
x_{K''}=\sum_{j=1}^\tau\rho_ju^j,
$$
where $u^j$ is $\big(\big(t-1)(k+|T^c|)-k'\big)$-sparse and
\begin{align*}
\sum_{j=1}^{\tau}&\rho_j=1,~0\leq\rho_j\leq1,j=1,\ldots,\tau\\
\text{supp}&(u^j)\subset \text{supp}(x_{K''})\\
\|u^j\|_1&=\|x_{K''}\|_1, \|u^j\|_{\infty}\leq \frac{\alpha}{t-1}.
\end{align*}
Hence,
\begin{align}\label{e6}
\|u^j\|_2\leq\|u^j\|_{\infty}\sqrt{\|u^j\|_0}\leq \frac{\alpha}{t-1} \sqrt{\big(t-1)(k+|T^c|)-k'}\leq\sqrt{\frac{k+|T^c|}{t-1}}\alpha.
\end{align}

We observe that
\begin{align}\label{e12}
\langle A(x_{K_0}&+x_{K'}),Ax\rangle\\
&=\frac{1}{4\delta}\sum_{j=1}^\tau\rho_j\Big(\|A\big((1+\delta)(x_{K_0}+x_{K'})+\delta
u^j\big)\|_2^2-\|A\big((1-\delta)(x_{K_0}+x_{K'})-\delta u^j\big)\|_2^2\Big).
\end{align}
Owing to
\begin{align*}
\|(1+\delta)(x_{K_0}+x_{K'})+\delta u^j\|_0&=\|(1-\delta)(x_{K_0}+x_{K'})-\delta u^j\|_0\\
&\leq (k+|T^c|)+k'+\big((t-1)(k+|T^c|)-k'\big)=t(k+|T^c|),
\end{align*}
therefore we take $\delta=\delta_{t(k+|T^c|)}$.

By $\|x_{K_0}+x_{K'}\|_0\leq (k+|T^c|)+k'\leq t(k+|T^c|)$,  we give a upper bound estimate for the left-hand side of (\ref{e12}) as follows
\begin{align*}
\langle A(x_{K_0}+x_{K'}),Ax\rangle \leq\|A(x_{K_0}+x_{K'})\|_2\|Ax\|_2\leq\sqrt{1+\delta}\|Ax\|_2\|x_{K_0}+x_{K'}\|_2
\end{align*}
or
\begin{align*}
\langle A(x_{K_0}+x_{K'}),Ax\rangle\leq\|x_{K_0}+x_{K'}\|_1\|A^*Ax\|_\infty\leq\sqrt{t(k+|T^c|)}\|A^*Ax\|_\infty\|x_{K_0}+x_{K'}\|_2.
\end{align*}

On the other hand, by restricted isometry property, the right-hand side of (\ref{e12}) is bounded from below by
\begin{align*}
\frac{1}{4\delta}&\sum_{j=1}^\tau\rho_j\Big((1-\delta)\big((1+\delta)^2\|x_{K_0}+x_{K'}\|_2^2+\delta^2\|u^j\|_2^2\big)
-(1+\delta)\big((1-\delta)^2\|x_{K_0}+x_{K'}\|_2^2+\delta^2 \|u^j\|_2^2\big)\Big)\\
&=\frac{1}{4\delta}\sum_{j=1}^N\rho_j\big(2\delta(1+\delta)(1-\delta)\|x_{K_0}+x_{K'}\|_2^2-2\delta^3\|u^j\|_2^2\big)\\
&\geq\frac{1-\delta^2}{2}\|x_{K_0}+x_{K'}\|_2^2-\frac{\delta^2}{2}\Big(\sqrt{\frac{k+|T^c|}{t-1}}\alpha\Big)^2\\
&=\frac{1-\delta^2}{2}\|x_{K_0}+x_{K'}\|_2^2-\frac{\delta^2}{2}\frac{\|x_{T\cap K^c}\|_1^2}{(t-1)(k+|T^c|)},
\end{align*}
where the last inequality follows by (\ref{e6}).

Combining the two bounds gives
\begin{align*}
\|x_{K_0}+x_{K'}\|_2^2-\frac{\delta^2}{(t-1)(1-\delta^2)}\Big(\frac{\|x_{T\cap K^c}\|_1}{\sqrt{k+|T^c|}}\Big)^2
\leq\frac{2\sqrt{1+\delta}}{1-\delta^2}\|Ax\|_2\|x_{K_0}+x_{K'}\|_2
\end{align*}
or
\begin{align*}
\|x_{K_0}+x_{K'}\|_2^2-\frac{\delta^2}{(t-1)(1-\delta^2)}\Big(\frac{\|x_{T\cap K^c}\|_1}{\sqrt{k+|T^c|}}\Big)^2
\leq\frac{2\sqrt{2(k+|T^c|)}}{1-\delta^2}\|A^*Ax\|_{\infty}\|x_{K_0}+x_{K'}\|_2,
\end{align*}
which is equivalent to
\begin{align*}
\Big(\|x_{K_0}+x_{K'}\|_2-\frac{\sqrt{1+\delta}}{1-\delta^2}\|Ax\|_2\Big)^2
\leq\frac{\delta^2}{(t-1)(1-\delta^2)}\Big(\frac{\|x_{T\cap K^c}\|_1}{\sqrt{k+|T^c|}}\Big)^2
+\frac{1+\delta}{(1-\delta^2)^2}\|Ax\|_2^2
\end{align*}
or
\begin{align*}
\Big(\|x_{K_0}+x_{K'}\|_2-\frac{\sqrt{2(k+|T^c|)}}{1-\delta^2}\|A^*Ax\|_{\infty}\Big)^2
\leq\frac{\delta^2}{(t-1)(1-\delta^2)}\Big(\frac{\|x_{T\cap K^c}\|_1}{\sqrt{k+|T^c|}}\Big)^2
+\frac{2(k+|T^c|)}{(1-\delta^2)^2}\|A^*Ax\|_2^2.
\end{align*}

As a result, we obtain
\begin{align*}
\|x_{K}\|_2\leq\|x_{K_0}\|_2\leq\|x_{K_0}+x_{K'}\|_2
&\leq\frac{2\sqrt{1+\delta}}{1-\delta^2}\|Ax\|_2+ \frac{\delta}{\sqrt{(t-1)(1-\delta^2)}}\frac{\sigma_k(x_T)_1}{\sqrt{k+|T^c|}}\\
&\leq\frac{2\sqrt{1+\delta}}{1-\delta^2}\|Ax\|_2+ \frac{\delta}{\sqrt{(t-1)(1-\delta^2)}}\frac{\sigma_k(x_T)_1}{\sqrt{k}}\\
&=:D_1\|Ax\|_2+\beta \frac{\sigma_k(x_T)_1}{\sqrt{k}}
\end{align*}
and
\begin{align*}
\|x_{K}\|_2\leq\|x_{K_0}\|_2\leq\|x_{K_0}+x_{K'}\|_2
&\leq\frac{2\sqrt{2(k+|T^c|)}}{1-\delta^2}\|A^*Ax\|_{\infty}+ \frac{\delta}{\sqrt{(t-1)(1-\delta^2)}}\frac{\sigma_k(x_T)_1}{\sqrt{k+|T^c|}}\\
&\leq\frac{2\sqrt{2(k+|T^c|)}}{1-\delta^2}\|A^*Ax\|_{\infty}+ \frac{\delta}{\sqrt{(t-1)(1-\delta^2)}}\frac{\sigma_k(x_T)_1}{\sqrt{k}}\\
&=:D_2\|A^*Ax\|_2+\beta \frac{\sigma_k(x_T)_1}{\sqrt{k}},
\end{align*}
which are the desired inequalities with
$$
\beta=\frac{\delta}{\sqrt{(t-1)(1-\delta^2)}}<1
$$
when $\delta_{t(k+|T^c|)}=\delta <\sqrt{\frac{t-1}{t}}$, and with
$$
D_1=\frac{2\sqrt{1+\delta}}{1-\delta^2},
D_2=\frac{2\sqrt{2(k+|T^c|)}}{1-\delta^2}.
$$
\end{proof}

From Theorems \ref{RIPSRP} and Theorem \ref{vectorstablerecovery}, we can obtain that under the condition $\delta_{t(k+|T^c|)}<\sqrt{\frac{t-1}{t}}$, the
stable recovery holds as follows.

\begin{theorem}\label{vectorstabilityrobustnessestimate}
Consider the signal model (\ref{systemequationsnoise}) with $\|z\|_2\leq\varepsilon$ and let $\hat{x}^{l_2}$ is the minimizer of (\ref{Truncatedq}) with
$\mathcal{B}=\mathcal{B}^{l_2}(\eta)$ defined in (\ref{lpboundednoise}) for some $\eta\geq\varepsilon$. If measurement matrix $A$ has restricted isometry property of order $t(k+|T^c|)$ with $\delta_{t(k+|T^c|)}<\sqrt{(t-1)/t}$ for some $t\geq4/3$, then
\begin{align*}
\|\hat{x}^{l_2}-x\|_2&\leq
\bigg(\max\Big\{\Big(\frac{|T^c|}{k}\Big),1\Big\}+1+\Big(\frac{|T|-k}{k}\Big)^{1/2}\bigg)\frac{D_1}{1-\beta}(\varepsilon+\eta)\nonumber\\
&\hspace*{12pt}+\bigg(\max\Big\{\Big(\frac{|T^c|}{k}\Big),1\Big\}+1+\Big(\frac{|T|-k}{k}\Big)^{1/2}\bigg)\frac{2\beta}{1-\beta}
k^{-1/2}\sigma_k(x_T)_1,
\end{align*}
where $D_1$ is the constant in above Theorem \ref{RIPSRP}.

Consider the signal model (\ref{systemequationsnoise}) with $\|A^*z\|_\infty\leq\varepsilon$ and let $\hat{x}^{DS}$ is the minimizer of (\ref{Truncatedq}) with $\mathcal{B}=\mathcal{B}^{DS}(\eta)$ defined in (\ref{Dantzigselectornoise}) for some $\eta\geq\varepsilon$. If measurement matrix $A$ has restricted isometry property of order $t(k+|T^c|)$ with $\delta_{t(k+|T^c|)}<\sqrt{(t-1)/t}$ for some $t\geq4/3$, then
\begin{align*}
\|\hat{x}^{l_2}-x\|_2&\leq
\bigg(\max\Big\{\Big(\frac{|T^c|}{k}\Big),1\Big\}+1+\Big(\frac{|T|-k}{k}\Big)^{1/2}\bigg)\frac{D_1}{1-\beta}(\varepsilon+\eta)\nonumber\\
&\hspace*{12pt}+\bigg(\max\Big\{\Big(\frac{|T^c|}{k}\Big),1\Big\}+1+\Big(\frac{|T|-k}{k}\Big)^{1/2}\bigg)\frac{2\beta}{1-\beta}
k^{-1/2}\sigma_k(x_T)_1,
\end{align*}
where $D_2$ is the constant in above Theorem \ref{RIPSRP}.
\end{theorem}

\begin{remark}
It should be point out that Theorem \ref{RIPSRP} and Theorem \ref{vectorstabilityrobustnessestimate} also hold for $1<t<4/3$ with the same proof. However the bound $\sqrt{(t-1)/t}$ is not sharp for $1<t<4/3$. See Cai and Zhang's \cite{CZ2014} for further discussions.
\end{remark}

\begin{remark}
We point out that if we take $T^c=\emptyset$, Theorem \ref{RIPSRP} implies that if  $A$ has restricted 2-isometry property of order $tk$ with $\delta_{tk}<\sqrt{(t-1)/t}$ for some $t\geq 4/3$ , then $A$ satisfies $(l_2,l_1)-l_2$ sparse approximation property of order $k$ certain constants $0<\beta<1$ and $D_1>0$,  and $(l_2,l_1)-\text{Dantzig~selector}$ sparse approximation property of order $k$ with certain constants $0<\beta<1$ and $D_2>0$, i.e $l_2$-robust null space property of order $k$ \cite{F2014}. This shows that the $(l_2,l_1)-l_2$ sparse approximation property of order $k$ and and $(l_2,l_1)-\text{Dantzig~selector}$ sparse approximation property of order $k$ is weaker than restricted 2-isometry property of order $tk$.
\end{remark}

Combining this result with \cite[Theorem 1.1]{S2011}, we get following robust estimate, which is same as \cite[Theorem 2.1]{CZ2014} but with different constants.

\begin{theorem}\label{VS-Cai-Zhang-2014}
Consider the signal model (\ref{systemequationsnoise}) with $\|z\|_2\leq\varepsilon$ and let $\hat{x}^{l_2}$ is the minimizer of (\ref{lqminimization}) with $\mathcal{B}=\mathcal{B}^{l_2}(\eta)$ defined in (\ref{lpboundednoise}) for some $\eta\geq\varepsilon$. If measurement matrix $A$ has restricted 2-isometry property of order $tk$ with $\delta_{tk}<\sqrt{(t-1)/t}$ for some $t\geq4/3$, then
\begin{align*}
\|\hat{x}^{l_2}-x\|_2&\leq
\frac{2\sqrt{1+\delta}\big(3\sqrt{(t-1)(1-\delta^2)}+\delta\big)}{(1-\delta^2)(\sqrt{(t-1)(1-\delta^2)}-\delta)}(\varepsilon+\eta)\\
&\hspace*{12pt}+\frac{\big(\sqrt{(t-1)(1-\delta^2)}+\delta\big)^2}{\sqrt{(t-1)(1-\delta^2)}\big(\sqrt{(t-1)(1-\delta^2)}-\delta\big)}\frac{2\sigma_k(x)_1}{\sqrt{k}}.
\end{align*}

Consider the signal model (\ref{systemequationsnoise}) with $\|A^*z\|_\infty\leq\varepsilon$ and let $\hat{x}^{DS}$ is the minimizer of (\ref{lqminimization}) with $\mathcal{B}=\mathcal{B}^{DS}(\eta)$ defined in (\ref{Dantzigselectornoise}) for some $\eta\geq\varepsilon$. If measurement matrix $A$ has restricted 2-isometry property of order $tk$ with $\delta_{tk}<\sqrt{(t-1)/t}$ for some $t\geq4/3$, then
\begin{align*}
\|\hat{x}^{DS}-x\|_2&\leq
\frac{2\sqrt{2k}\big(3\sqrt{(t-1)(1-\delta^2)}+\delta\big)}{(1-\delta^2)(\sqrt{(t-1)(1-\delta^2)}-\delta)}(\varepsilon+\eta)\\
&\hspace*{12pt}+\frac{\big(\sqrt{(t-1)(1-\delta^2)}+\delta\big)^2}{\sqrt{(t-1)(1-\delta^2)}\big(\sqrt{(t-1)(1-\delta^2)}-\delta\big)}\frac{2\sigma_k(x)_1}{\sqrt{k}}.
\end{align*}
\end{theorem}

\begin{remark}
It also should be point out that Zhang and Li \cite{ZL2017} gave a complete answer to the
conjecture on restricted isometry property constants $\delta_{tk}~(0<t<4/3)$ which was proposed by Cai and Zhang \cite{CZ2014}.
They showed that when $0<t<4/3$, the condition $\delta_{tk}<t/(4-t)$ is sufficient to guarantee the exact recovery for all
$k$-sparse signals in the noiseless case via the constrained $\ell_q$-norm minimization for $0<q\leq 1$ and these bounds are also sharp.
\end{remark}

We propose the following conjecture.

\begin{conjecture}\label{SRPweakerRIP}
If $A$ has restricted 2-isometry property of order $tk$ with $\delta_{tk}<(4-t)/t$ for some $t\in(0,4/3)$ , then $A$ satisfies $(l_2,l_1)-l_2$ sparse approximation property of order $k$ with certain constants $0<\gamma<1$ and $0<D_3<\infty$,  and  $(l_2,l_1)-\text{Dantzig~selector}$ sparse approximation property of order $k$ with certain constants $0<\gamma<1$ and $0<D_4<\infty$.
\end{conjecture}

\begin{remark}
We should point out that the above Theorem \ref{RIPSRP}, Theorem \ref{vectorstabilityrobustnessestimate} Theorem \ref{VS-Cai-Zhang-2014} and Theorem \ref{k2ksparsevector} also hold for matrix case.
\end{remark}

\section{Conclusions and Discussion \label{s4}}
\hskip\parindent

To solve truncated norm minimization (\ref{Truncatedq}) or (\ref{MatrixTruncatedq}) with noise, we introduce truncated sparse approximation property, and obtain the stable recovery of signals and matrices under the truncated sparse approximation property (see Theorem \ref{vectorstablerecovery} and Theorem \ref{Matrixstablerecovery}). When we investigate the relationship between the truncated sparse approximation property and  restricted isometry property, we find that if a measurement matrix satisfies truncated sparse approximation property of order $k$, then the first inequality in restricted isometry property of order $k$ and of order $2k$ hold for certain different constants $\delta_{k}$ and $\delta_{2k}$, respectively (see Theorem \ref{k2ksparsevector}). And we also find that the truncated sparse approximation property of order $k$ can be deduced from the restricted isometry property with $\delta_{tk}<\sqrt{(t-1)/t}$ for some $t\geq4/3$ (see Theorem \ref{RIPSRP}). And especially when $T^c=\emptyset$-the classic case, it shows us that the $(l_2,l_1)-l_2$ sparse approximation property of order $k$ and and $(l_2,l_1)-\text{Dantzig~selector}$ sparse approximation property of order $k$ is weaker than restricted 2-isometry property of order $tk$. And combining this result with \cite[Theorem 1.1]{S2011}, we get the stable recovery of signals and matrices under the restricted isometry property of order $tk$ with $\delta_{tk}<\sqrt{(t-1)/t}$ for some $t\geq 4/3$ (see Theorem \ref{VS-Cai-Zhang-2014}), which is the same as \cite[Theorem 2.1]{CZ2014} but with different constants. Therefore, these results in our paper may guide the practitioners to apply robust null space property in compressed sensing, low-rank matrix recovery and sparse phase retrieval, especially in truncated norm minimization problem.

However, Zhang and Li \cite{ZL2017} showed that for $0<t<4/3$, the condition $\delta_{tk}<t/(4-t)$ is sufficient to guarantee the exact recovery for all $k$-sparse signals in the noiseless case via the constrained $l_q$-norm minimization and these bounds are also sharp. Whether the truncated sparse approximation property of order $k$ can be deduced from the restricted isometry property with $\delta_{tk}<t/(4-t)$ for some $0<t<4/3$? But our method in Theorem \ref{RIPSRP} may not fit for this question. This is one direction of our future research (see Conjecture \ref{SRPweakerRIP}).
And another question is there exist't any algorithm for computing truncated $\ell_q$ minimization (\ref{Truncatedq}) with $\mathcal{B}=\mathcal{B}^{DS}(\eta)$ or $\mathcal{B}=\mathcal{B}^{l_q}(\eta)$ for $0<q\leq 1$ (see Remark \ref{Algorithm})? This is anther direction of our future research.


\textbf{Acknowledgement}: Wengu Chen is supported by National Natural Science Foundation of China (No. 11371183).




\bigskip

\noindent Wengu Chen

\medskip

\noindent
Institute of Applied Physics and Computational Mathematics, Beijing,
100088, People's Republic of China\\
\smallskip
\noindent{\it E-mail addresses}: \texttt{chenwg@iapcm.ac.cn} (W. Chen)

\bigskip

\noindent Peng Li

\medskip

\noindent Graduate School, China Academy of Engineering Physics, Beijing,
100088, People's Republic of China\\
\smallskip
\noindent{\it E-mail addresses}: \texttt{lipeng16@gscaep.ac.cn} (P. Li)

\end{document}